\documentclass[conference,fontsize=10pt]{IEEEtran}
\IEEEoverridecommandlockouts
\usepackage[utf8]{inputenc}

\usepackage{amsmath}
\usepackage[standard,thmmarks,amsmath]{ntheorem}
\usepackage{dsfont}
\usepackage[dvipsnames, table]{xcolor}
\usepackage{tikz}
\usetikzlibrary{automata,positioning}

\usepackage{paralist}

\usepackage{mathpartir}
\usepackage{prftree}
\usepackage{graphicx}
\usepackage{mathtools}
\usepackage{makecell}
\usepackage{color, colortbl}
\definecolor{LightCyan}{rgb}{0.88,1,1}

\usepackage{enumitem}
\usepackage{subcaption}

\usepackage{comment}

\usepackage{booktabs}

\usepackage[hidelinks]{hyperref}

\usepackage{cleveref}
\crefname{definition}{Def.}{Defs.}
\crefname{proposition}{Prop.}{Props.}
\crefname{algorithm}{Alg.}{Algs.}
\crefname{table}{Tab.}{Tabs.}
\crefname{section}{Sec.}{Sec.}
\crefname{figure}{Fig.}{Fig.}
\crefname{example}{Ex.}{Ex.}
\crefname{lemma}{Lemma}{Lemma}

\usepackage{algorithm}
\usepackage[noend]{algpseudocode}
\algdef{SE}[DOWHILE]{Do}{doWhile}{\algorithmicdo \ \{ }[1]{\} \algorithmicwhile\ #1}\algnewcommand\algorithmicforeach{\textbf{for each}}
\algdef{S}[FOR]{ForEach}[1]{\algorithmicforeach\ #1\ \algorithmicdo}
\algnewcommand\algorithmicmatch{\textbf{match}}
\algnewcommand\algorithmicwith{\textbf{with}}
\algdef{SE}[MATCH]{Match}{EndMatch}[1]{\algorithmicmatch\ #1}{\algorithmicend\ \algorithmicmatch}\algdef{SE}[MCASE]{MCase}{EndMCase}[1]{\algorithmicwith\ #1:}{\algorithmicend\ }\algtext*{EndMCase}\algtext*{EndMatch}

\algnewcommand\algorithmicswitch{\textbf{switch}}
\algnewcommand\algorithmiccase{\textbf{case}}
\algdef{SE}[SWITCH]{Switch}{EndSwitch}[1]{\algorithmicswitch\ #1\ \algorithmicdo}{\algorithmicend\ \algorithmicswitch}\algdef{SE}[CASE]{Case}{EndCase}[1]{\algorithmiccase\ #1:}{\algorithmicend\ \algorithmiccase}\algtext*{EndSwitch}\algtext*{EndCase}

\usepackage{xspace}
\usepackage{url}

\crefname{problem}{Problem}{Problems}

\algrenewcommand\alglinenumber[1]{\scriptsize #1:}

\newcommand{\limp}{\Rightarrow}
\newcommand{\nextsplit}[0]{{\sf NextSplit}\xspace}

\newcommand{\fixed}[0]{{\sf IsFixed}\xspace}

\newcommand{\specsms}[0]{\textsc{specSMS}\xspace}
\newcommand{\eqdef}[0]{\triangleq}
\newcommand{\sha}[0]{\operatorname{SHA-1}}
\newcommand{\pdr}[0]{IC3/PDR\xspace}
\newcommand{\sms}[0]{SMS\xspace}
\newcommand{\php}[0]{PHP}
\newcommand{\fa}[0]{\textsc{sat}}
\newcommand{\conflict}[0]{\textsc{c}}
\newcommand{\analyzefinal}[0]{\textsf{AnalyzeFinal}\xspace}

\newcommand{\emptylist}{[\hphantom{a}]}
\newcommand{\nulllit}{\dagger}
\newcommand{\dontcare}{\ast}
\newcommand{\db}{\Phi}
\newcommand{\mdl}{\mathcal{M}}

\newcommand{\cls}[0]{\ensuremath{\mathit{cls}}\xspace}
\newcommand{\idx}[0]{\ensuremath{\mathit{idx}}\xspace}
\newcommand{\rup}[0]{\ensuremath{\text{rup}}\xspace}
\newcommand{\cp}[0]{\ensuremath{\text{cp}}\xspace}
\newcommand{\del}[0]{\ensuremath{\text{del}}\xspace}
\newcommand{\type}[0]{\ensuremath{\mathit{type}}\xspace}

\newcommand{\actcls}[0]{\ensuremath{\mathit{act\_clauses}}\xspace}
\newcommand{\checkrup}[0]{\textsc{chk\_rup}\xspace}
\newcommand{\trup}{\ensuremath{\vdash_{\mathit{UP}}}\xspace}
\newcommand{\placeholder}{\text{ext}}
\newcommand{\decidemode}[0]{D}
\newcommand{\cc}[0]{C}

\newcommand{\unresolvable}[0]{\textsc{REFINE}\xspace}
\newcommand{\itp}[0]{\ensuremath{\mathit{itp}}\xspace}
\newcommand{\allasserted}[0]{\textsc{asserted}\xspace}
\newcommand{\invars}[0]{\ensuremath{\mathit{in}}\xspace}
\newcommand{\outvars}[0]{\ensuremath{\mathit{out}}\xspace}
\newcommand{\msg}[0]{\ensuremath{\mathit{msg}}\xspace}
\newcommand{\id}[0]{\ensuremath{\sharp}}

\newcommand{\mains}[0]{{\sf m}\xspace}
\newcommand{\secs}[0]{{\sf s}\xspace}

\newcommand{\la}{\ensuremath{\mathit{la}}\xspace}
\newcommand{\lb}{\ensuremath{\mathit{lb}}\xspace}

\newcommand{\chks}[0]{\ensuremath{\mathit{chks}}\xspace}

\newcommand{\sola}[0]{\ensuremath{S_{\mains}}\xspace}
\newcommand{\solb}[0]{\ensuremath{S_{\secs}}\xspace}
\newcommand{\SM}[0]{\ensuremath{\mathit{SM}}\xspace}
 \title{Speculative SAT Modulo SAT}
\author{
\IEEEauthorblockN{Hari Govind V K}
\IEEEauthorblockA{\textit{University of Waterloo}\\
Waterloo, Canada \\
hgvedira@uwaterloo.ca}
\and
\IEEEauthorblockN{Isabel Garcia-Contreras}
\IEEEauthorblockA{\textit{University of Waterloo}\\
Waterloo, Canada \\
igarciac@uwaterloo.ca}
\and
\IEEEauthorblockN{Sharon Shoham}
\IEEEauthorblockA{\textit{Tel-Aviv University}\\
Tel-Aviv, Israel \\
sharon.shoham@cs.tau.ac.il}
\and
\IEEEauthorblockN{Arie Gurfinkel}
\IEEEauthorblockA{\textit{University of Waterloo}\\
Waterloo, Canada \\
arie.gurfinkel@uwaterloo.ca}
}

\begin{document}
\maketitle
\thispagestyle{plain}
\pagestyle{plain}

\begin{abstract}
 State-of-the-art model-checking algorithms like \pdr are based on uni-directional
 modular SAT solving for finding and/or blocking counterexamples. Modular SAT-solvers divide a SAT-query into multiple sub-queries, each solved by a separate SAT-solver (called a module), and propagate information (lemmas, proof obligations, blocked clauses, etc.) between modules. While modular solving is key to \pdr, it is obviously not as effective as monolithic solving, especially when individual sub-queries are harder to solve than the combined query. This is partially addressed in SAT modulo SAT~(\sms) by propagating unit literals back and forth between the modules and using information from one module to
 simplify the sub-query in another module as soon as possible (i.e., before the satisfiability of any sub-query is established). However, bi-directionality of \sms is limited because of the strict order between decisions and propagation -- only one module is allowed to make decisions, until its sub-query is SAT. In this paper, we propose a generalization of \sms, called \specsms, that \emph{speculates} decisions between modules. This makes it bi-directional -- decisions are made in multiple modules, and learned clauses are exchanged in both directions. We further extend DRUP proofs and interpolation, these are useful in model checking, to \specsms. We have implemented \specsms in Z3 and show that it performs exponentially better on a series of benchmarks that are provably hard for \sms.

\end{abstract}
 \section{Introduction}
\label{sec:intro}

\pdr~\cite{DBLP:conf/vmcai/Bradley11} is an efficient SAT-based Model Checking
algorithm. Among many other innovations in \pdr is the concept of
a modular SAT-solver that divides a formula into multiple \emph{frames} and
each frame is solved by an individual SAT solver. The solvers communicate by
exchanging proof obligations (i.e., satisfying assignments) and lemmas (i.e.,
learned clauses).

While modular reasoning in \pdr is very efficient for a Model Checker, it is not as efficient as a
classical monolithic SAT-solver. This is not surprising since modularity restricts the solver to colorable
refutations~\cite{DBLP:conf/fmcad/GurfinkelV14}, which are, in the worst case, exponentially bigger than unrestricted refutations. On the positive side, \pdr's modular SAT-solving makes interpolation trivial, and enables generalizations of proof obligations and inductive generalization of lemmas -- both are key to the success of \pdr. 

This motivates the study of modular SAT-solving, initiated by \sms~\cite{SMS}. 
Our strategic vision is that our study will contribute to improvements in \pdr. However, in this paper, we focus on modular SAT-solving in isolation.

In modular SAT-solving, multiple solvers interact to check satisfiability of a partitioned CNF formula, where each part of the formula is solved by one of the solvers. In this paper, for simplicity, we consider the case of
two solvers $\langle \solb, \sola \rangle$ checking satisfiability of a formula pair $\langle \db_\secs, \db_\mains \rangle$. 
$\sola$ is a \emph{main} solver
and $\solb$ is a \emph{secondary} solver. In the notation, the solvers are written
right-to-left to align with \pdr, where the main solver is used for frame~1
and the secondary solver is used for frame~0.

When viewed as a modular SAT-solver, \pdr is uni-directional. First, $\sola$
finds a satisfying assignment $\sigma$ to $\db_{\mains}$ and only then, $\solb$
extends $\sigma$ to an assignment for $\db_{\secs}$. Learned clauses, called
\emph{lemmas} in \pdr, are only shared (or copied) from the secondary solver
$\solb$ to the main solver $\sola$.

SAT Modulo SAT (\sms)~\cite{SMS} is a modular SAT-solver that extends \pdr by allowing inter-modular unit
propagation and conflict analysis: whenever an interface literal is placed on a trail of any solver,
it is shared with the other solver and both solvers run unit propagation, exchanging unit literals. 
This makes
modular SAT-solving in \sms bi-directional as information flows in both directions
between the solvers.
Bi-directional reasoning can simplify proofs, but it significantly complicates conflict analysis.
To manage conflict analysis, \sms does not allow the secondary solver $\solb$ to make any
decisions before the main solver $\sola$ is able to find a complete assignment
to its clauses. As a result, learned clauses are either local to each solver, or
flow only from $\solb$ to $\sola$, restricting the structure of refutations similarly to \pdr.

Both \pdr and \sms require $\sola$ to find a complete satisfying assignment to $\db_{\mains}$ before the solving is continued in $\solb$. This is problematic since $\db_{\mains}$ might be hard to satisfy, causing them to get stuck in $\db_{\mains}$,  even if considering both formulas together quickly reveals the (un)satisfiability of $\langle \db_\secs, \db_\mains \rangle$.

In this paper, we introduce \specsms{} --- a modular SAT-solver that employs a truly bi-directional reasoning. 
\specsms builds on \sms, while facilitating deeper communication between
the modules by 
\begin{inparaenum}[(1)] 
  \item allowing learnt clauses to flow in both directions, and 
  \item letting the two solvers interleave their decisions.
\end{inparaenum}
The key challenge is in the
adaptation of conflict analysis to properly handle the case of a conflict that depends on
decisions over local variables of both solvers. Such a conflict cannot be
explained to either one of the solvers using only interface clauses (i.e., clauses over interface variables).
It may, therefore, require backtracking
the search without learning any conflict clauses.
To address this challenge, \specsms uses \emph{speculation}, which tames decisions of the secondary solver that are interleaved with decisions of the main solver. If the secondary solver satisfies all of its clauses during speculation, a \emph{validation} phase is employed, where the main solver attempts to extend the assignment to satisfy its unassigned clauses. If speculation leads to a conflict which depends on local decisions of both solvers, \emph{refinement} is employed to resolve the conflict. Refinement ensures progress even if no conflict clause can be learnt. With these ingredients, we show that
\specsms is sound and complete (i.e., always terminates).

To certify \specsms's result when it determines that a formula is unsatisfiabile, we extract a \emph{modular} clausal proof from its execution. To this end, we extend DRUP proofs~\cite{DBLP:conf/fmcad/HeuleHW13} to account for modular reasoning, and devise a procedure for trimming modular proofs. Such proofs are applicable both to \specsms and to \sms.
Finally, we propose an interpolation algorithm that extracts an interpolant~\cite{craig} from a modular proof. Since clauses are propagated between the solvers in both directions, the extracted interpolants have the shape $\bigwedge_i (C_i \limp \cls_i$), where $C_i$ are conjunctions of clauses and each $\cls_i$ is a  clause.

Original \sms is implemented on top of MiniSAT. For this
paper, we implemented both \sms and \specsms in Z3~\cite{Z3}, using the extendable SAT-solver interface of Z3. 
Thanks to its bi-directional reasoning, \specsms is able to efficiently solve both sat and unsat formulas that are provably hard for existing modular SAT-solvers, provided that speculation is performed at the right time. \specsms relies on a user-provided callback to decide when to speculate. Developing good heuristics for speculation is a topic of future work.

In summary, we make the following contributions:
\begin{inparaenum}[(i)]
\item the \specsms algorithm that leverages bi-directional modular reasoning (\cref{sec:specsms});
\item modular DRUP proofs for \specsms  (\cref{sec:proofs});
\item proof-based interpolation algorithm;
\item user interface to guide speculation and search (\cref{sec:api}); and 
\item implementation and validation (\cref{sec:impl}).
\end{inparaenum}

 \section{Motivating examples\label{sec:examples}}
In this section, we discuss two examples in which both \pdr-style uni-directional reasoning and \sms-style shallow bi-directional reasoning are ineffective.
The examples illustrate why existing modular reasoning gets stuck.
To better convey our intuition, we present our problems at word level using bit-vector variables directly, without explicitly converting them to propositional variables. 

\begin{example}
Consider the following modular sat query: $\langle\varphi_{\invars}, \varphi_{\sha}\rangle$, where $\varphi_{\invars} \eqdef (\invars = \invars_1) \lor (\invars = \invars_2)$, $\invars$ is a 512-bit vector, $\invars_1$, $\invars_2$ are 512-bit values,  $\varphi_{\sha} \eqdef (\sha_{\mathit{circ}}(\invars) = \sha_{\invars_1}$), $\sha_{\mathit{circ}}(\invars)$ is a circuit that computes $\sha$ of $\invars$, and $\sha_{\invars_1}$ is the 20 byte $\sha$ message digest of $\invars_1$. 

Checking the satisfiability of $\varphi_{\invars} \land \varphi_{\sha}$ is easy because it contains both the output and the input of the $\sha$ circuit. However, existing modular SAT-solvers attempt to solve the problem starting by finding a complete satisfying assignment to $\varphi_{\sha}$. This is essentially the problem of inverting the $\sha$ function, which is known to be very hard for a SAT-solver. 
The improvements in \sms do not help. There are no unit clauses in $\varphi_{\invars}$. 

On the other hand, \specsms proceeds as follows: 
\begin{inparaenum}[(1)]
\item when checking satisfiability of $\varphi_{\sha}$, it decides to speculate, 
\item it starts checking satisfiability of $\varphi_{\invars}$, branches on variables $\invars$, finds an assignment $\sigma$ to $\invars$ and unit propagates $\sigma$ to $\varphi_{\sha}$, 
\item if there is a conflict in $\varphi_{\sha}$, it learns the conflict clause $\invars \neq \invars_2$, and 
\item it terminates with a satisfying assignment $\invars = \invars_1$. 
\end{inparaenum}
Speculation in step (1) is what differentiates \specsms from \pdr and \sms. The specifics of when exactly \specsms speculates is guided by the user, as explained in \cref{sec:api}.
\end{example}

\begin{example}
Speculation is desirable for unsatisfiable formulas as well. Consider the modular sat query $\langle\varphi_{+},\varphi_{-}\rangle$, where $\varphi_{+} \eqdef (a < 0\limp x) \land (a \geq 0\limp x) \land \php^1_{32}$ and $\varphi_{-} \eqdef (b < 0\limp \neg x) \land (b \geq 0\limp \neg x) \land \php^2_{32}$. Here, $a$ and $b$ are 32-wide bitvectors and local to the respective modules. $\php_{32}$ encodes the problem of fitting $32$ pigeons into $31$ holes and $\php^1_{32}$ and $\php^2_{32}$ denote a partitioning of $\php_{32}$ into 2 problems such that both formulas contain all variables. The modular problem $\langle \varphi_{+}, \varphi_{-}\rangle$ is unsatisfiable, $x$ and $\php^1_{32}$ being two possible interpolants. \pdr and \sms only find the second interpolant. This is because, all satisfying assignments to $\varphi_{-}$ immediately produce a conflict in $\php^1_{32}$ part of $\varphi_{+}$.  However, learning an interpolant containing $x$ requires searching (i.e., deciding) in both $\varphi_{+}$ and $\varphi_{-}$. \specsms, with proper guidance, solves the problem by (1)~deciding on all $b$ variables in $\varphi_{-}$, (2)~switching to speculation, (3)~branching on all $a$ variables in $\varphi_{+}$ to hit a conflict in $x$, and, finally (4)~learning the conflict clause $x$. 
\end{example}

These examples highlight the need to speculate
 while doing modular reasoning. While speculation by itself is quite powerful, it requires proper guidance to be effective. We explain how the user can provide such a guidance in \cref{sec:api}.

 \section{Speculative SAT Modulo SAT}
\label{sec:specsms}
This section presents \specsms{} --- a modular bi-directional SAT algorithm. For
simplicity, we restrict our attention to the case of two modules. However,
the algorithm easily generalizes to any sequence of modules.

\subsection{Sat Modulo Sat}
We assume that the reader has some familiarity with internals of a MiniSAT-like
SAT solver~\cite{miniSAT} and with \sms~\cite{SMS}. We give a brief background
on \sms, highlighting some of the key aspects.
\sms decides satisfiability of a partitioned CNF formula $\langle \db_\secs,
\db_\mains\rangle$ with a set of shared interface variables $I$. It uses two
modules $\langle \solb, \sola \rangle$, where $\sola$ is a \emph{main} module
used to solve $\db_{\mains}$, and $\solb$ is a \emph{secondary} module to solve
$\db_{\secs}$. Each module is a SAT solver (with a slightly extended interface,
as described in this section). We refer to them as \emph{modules} or
\emph{solvers}, interchangeably. 
Each solver has its own clause
database (initialized with $\db_i$ for $i \in \{\mains, \secs\}$), and a trail
of literals, just as a regular SAT solver. The solvers keep their decision
levels in sync. Whenever a decision is made in one solver, the decision level of
the other solver is incremented as well (adding a \emph{null} literal to its
trail if necessary). Whenever one solver back-jumps to level $i$, the other
solver back-jumps to level $i$ as well. Assignments to interface variables are
shared between the solvers: whenever such a literal is added to the trail of one
solver (either as a decision or due to propagation), it is also added to the
trail of the other solver. \sms requires that $\solb$ does not make any
decisions, until $\sola$ finds a satisfying assignment to its clauses.

\paragraph*{Inter-modular propagation and conflict analysis} The two key features of SMS are inter-modular unit
propagation (called \textsc{PropagateAll} in~\cite{SMS}) and the corresponding
inter-modular conflict analysis. In \textsc{PropagateAll}, whenever an interface
literal is added to the trail of one solver, it is added to the trail of the
other, and both solvers run unit propagation. Whenever a unit literal $\ell$ is
copied from the trail of one solver to the other, the \texttt{reason} for $\ell$
in the destination solver is marked using a marker $\placeholder$. This indicates that
the justification for the unit is external to the destination
solver\footnote{This is similar to theory propagation in SMT solvers.}.
Propagation continues until either there are no more units to propagate or one
of the solvers hits a conflict.

Conflict analysis in \sms is extended to account for units with no reason clauses. If
such a literal $\ell$ is used in conflict analysis, its reason is obtained by
using $\analyzefinal(\ell)$ on the other solver to compute a clause $(s\limp
\ell)$ over the interface literals. This clause is copied to the requesting
solver and is used as the missing reason. Multiple such clauses can be copied
(or learned) during analysis of a single conflict clause -- one clause for each
literal in the conflict that is assigned by the other solver. 

In \sms, it is crucial that $\analyzefinal(\ell)$ always succeeds to generate a
reason clause over the interface variables. This is ensured by only calling
$\analyzefinal(\ell)$ in the $\solb$ solver on literals that were added to the
trail when $\solb$ was not yet making decisions. 
This can happen in one of two scenarios: either $\sola$ hits a conflict 
due to literals propagated from $\solb$, in which case $\analyzefinal$
is invoked in $\solb$ on each literal marked $\placeholder$ in $\sola$ that is involved in the conflict resolution to obtain its \texttt{reason}; 
or $\solb$ hits a conflict during unit propagation, in which case it invokes $\analyzefinal$
to obtain a conflict clause over the interface variables that blocks the partial assignment of $\sola$.
In both cases, new reason
clauses are always copied from $\solb$ to $\sola$. We refer the reader
to~\cite{SMS} for the pseudo-code of the above inter-modular procedures for
details.

\subsection{Speculative Sat Modulo Sat}
\specsms extends \sms~\cite{SMS} by a combination of \emph{speculation},
\emph{refinement}, and \emph{validation}. During the search in the main solver
$\sola$, \specsms non-deterministically \emph{speculates} by allowing
the secondary solver $\solb$ to extend the current partial assignment of
$\db_{\mains}$ to a satisfying assignment of $\db_{\secs}$. If $\solb$ is
unsuccessful (i.e., hits a conflict), and the conflict depends on a combination
of a \emph{local}
decision of $\sola$ with some decision of $\solb$, then the search reverts to
$\sola$ and its partial assignment is \emph{refined} by forcing $\sola$ to
decide on an interface literal from the conflict. On the other hand, if $\solb$
is successful, solving switches to the main solver $\sola$ that \emph{validates}
the current partial assignment by extending it to all of its clauses. This
either succeeds (meaning, $\langle \db_{\secs}, \db_{\mains} \rangle$ is sat),
or fails and another \emph{refinement} is initiated. Note that the two sub-cases where
$\solb$ is unsuccessful but the reason for the conflict is either local to
$\solb$ or local to $\sola$ are handled as in \sms.

\paragraph*{Search modes}
\specsms controls the behavior of the solvers and their interaction
through \emph{search modes}. Each solver can be in one of the following search
modes: Decide, Propagate, and Finished. In Decide, written $\decidemode^{i}$,
the solver treats all decisions below level $i$ as assumptions and is allowed to
both make decisions and do unit propagation. In Propagate, written $P$, the
solver makes no decisions, but does unit propagation whenever new literals are
added to its trail. In Finished, written $F$, the clause database of the solver
is satisfied; the solver neither makes decisions nor propagates unit literals.

The pair of search modes of both modules is called the \emph{state} of \specsms, where we add a unique state called \textit{unsat} for the case when the combination of the modules is known to be unsatisfiable.
The possible states and transitions of \specsms are shown in \cref{fig:smsfa}.
States $\mathit{unsat}$ and $\langle F,F \rangle$ are two final states, corresponding to unsat and sat, respectively. In all
other states, exactly one of the solvers is in a state $\decidemode^{i}$. We
refer to this solver as \emph{active}. The part of the transition system
highlighted in yellow correspond to \sms, and the green part includes the states
and transitions that are unique to \specsms. 

\usetikzlibrary {arrows.meta}
\begin{figure}
\centering
\pgfdeclarelayer{background}
\pgfsetlayers{background,main}
\begin{tikzpicture}[node distance = 1cm, auto, every state/.style={minimum size=1.2cm}]
  \node[state,initial] (q_0) {$\langle P, \decidemode^0\rangle$};
  \node[state] (q_1) [above=1.3cm of q_0] {$\langle \decidemode^i, P\rangle$};
   \node[state] (q_2) [right=of q_1] {$\langle F, \decidemode^{M}\rangle$};
   \node[state, accepting] (q_b) [right=1.3cmof q_0] {$\textit{unsat}$};
  \node[state] (q_3) [below=of q_0] {$\langle \decidemode^N, F\rangle$};
  \node[state, accepting] (q_4) [right=of q_3] {$\langle F, F\rangle$};
\begin{pgfonlayer}{background}
       \path (q_0.west |- q_0.north)+(-0.2,0.4) node (a) {};
        \path (q_4.east |- q_4.south)+(+0.4,-0.5) node (b) {};
        \path[fill=yellow!20,rounded corners] (a) rectangle (b) node [above left]{\sms and \specsms};

        \path (q_1.west |- q_1.north)+(-0.2,0.5) node (c) {};
        \path (q_2.east |- q_2.south)+(+0.15,-0.8) node (d) {};
        \path[fill=green!20,rounded corners]
            (c) rectangle (d) node [near end, above = 1.6cm]{only \specsms};
\end{pgfonlayer}
  \path[-{Stealth[length=2mm]}]
  (q_0) edge node {\scriptsize$\sola: \conflict\,@\,0$} (q_b)
        edge [bend left] node [midway] {$\scriptstyle\bigstar$}   (q_1)
        edge [bend left] node [sloped, above, rotate=0] {\scriptsize$\sola: \fa$} (q_3)
  (q_1) edge  node {\scriptsize$\solb: \fa$} (q_2)
        edge [bend left] node [sloped, above] {\scriptsize$\solb: \conflict\,@\,{\scriptscriptstyle \leq}\!i$} node [sloped, below] {\scriptsize \unresolvable} (q_0)
  (q_2) edge  node [sloped, below] {\scriptsize \unresolvable} (q_0.30) edge  [bend left]  node [sloped, near start, below] {\scriptsize$\sola: \fa$} (q_4)
  (q_3) edge [bend left] node [sloped, above] {\scriptsize$\solb: \conflict\,@\,{\scriptscriptstyle \leq}\!N$} (q_0)
        edge node {\scriptsize$\solb: \fa$}  (q_4);
\end{tikzpicture}
\caption{State transitions of \specsms. A state $\langle P,
  \decidemode^0\rangle$ means that the secondary solver $\solb$ is in propagate
  mode and the main solver $\sola$ is in solve mode. Each edge is guarded with a
  condition. The condition $\sola : \fa$ means that $\sola$ found a full
  satisfying assignment to $\db_{\mains}$. The condition $\sola:
  \conflict\,@\,{\scriptscriptstyle \leq}\!j$ means that $\sola$ hit a conflict
  at a decision level below $j$. The four states in yellow corresponds to \sms;
  two states in green are unique to \specsms.}
\label{fig:smsfa}
\vspace{-0.28in}
\end{figure}

\paragraph*{Normal execution with bi-directional propagation}
\specsms starts in the state $\langle P, \decidemode^0\rangle$, with the main
solver being active. In this state, it can proceed like \sms by staying in the
yellow region of \cref{fig:smsfa}. We call this \emph{normal execution with
  bi-directional propagation}, since (only) unit propagation goes between solvers.

\paragraph*{Speculation}
What sets \specsms apart is speculation: at any non-deterministically chosen
decision level $i$, \specsms can pause deciding on the main solver and activate the
secondary solver~(i.e., transition to state $\langle \decidemode^i, P\rangle$).
During speculation, only the secondary solver makes decisions. Since
the main solver does not have a full satisfying assignment to its clauses, the
secondary solver propagates assignments to the main solver and vice-versa.

Speculation terminates when the secondary solver $\solb$ either:
(1)~hits a conflict that cannot be resolved by inter-modular conflict analysis;
(2)~hits a conflict below decision level $i$; or (3)~finds a satisfying
assignment to $\db_{\secs}$.

Case (1) is most interesting, and is what makes \specsms differ from \sms. 
Note
that a conflict clause is not resolved by inter-modular conflict analysis
only if it depends on an external literal on the trail of $\solb$ that cannot be
explained by an interface clause from $\sola$. This is possible when both
$\sola$ and $\solb$ have partial assignments during speculation. So the conflict
might depend on the \emph{local} decisions of $\sola$. This cannot be
communicated to $\solb$ using only interface variables.

\paragraph*{Refinement} In \specsms, this is handled by modifying the \textsc{Reason} method in the
solvers to fail (i.e., return $\placeholder$) whenever $\analyzefinal$ returns a
non-interface clause. Additionally, the literal on which $\analyzefinal$ failed
is recorded in a global variable $\mathit{refineLit}$. This is shown in
\cref{alg:reason}. The inter-modular conflict analysis is
modified to exit early whenever \textsc{Reason} fails to produce a
justification. At this point, \specsms exits speculation, returns to the initial state
$\langle P, \decidemode^0 \rangle$, both solvers back-jump to decision level $i$ at which speculation was initiated, and  $\sola$ is forced to decide on $\mathit{refineLit}$.

We call this transition a \emph{refinement} because the partial assignment of
the main solver $\sola$ (which we view as an \emph{abstraction}) is updated
(a.k.a., refined) based on the information that was not available to it (namely,
a conflict with a set of decisions in the secondary solver $\solb$). Since
$\mathit{refineLit}$ was not decided on in $\sola$ prior to speculation,
deciding on it is a new decision that ensures progress in $\sola$. The next
speculation is possible only under strictly more decisions in $\sola$ than
before, or when $\sola$ back-jumps and flips an earlier decision.

We illustrate the refinement process on a simple example:
\begin{example}
  Consider the query $\langle\db_{\secs}, \db_{\mains}\rangle$ with: 
  \\[2mm]
  \begin{minipage}{.35\linewidth}
$\db_\secs(i,j,k,z)$:
\begin{align*}
  & \overline{z} \lor \overline{i} \tag{3} \\
& i \lor j \lor \overline{k} \tag{4}
\end{align*}
 \end{minipage}
\hfill \vline \hfill 
  \begin{minipage}{.35\linewidth}
      $\db_\mains(a,i,j,k)$:
    \begin{align*}
      & \overline{a} \lor i \lor \overline{j} \tag{1} \\
      & j \lor k \tag{2}
    \end{align*}
  \end{minipage}
  \\[1mm]
  First, \sola decides $a$ (at level 1), which causes no propagations. Then,
  \specsms enters speculative mode, transitions to $\langle D^1, P\rangle$ and starts making decisions in \solb. \solb
  decides $z$ and calls \textsc{PropagateAll}. Afterwards, 
  the trails for \sola and \solb are as follows:

\begin{center}
\begin{tabular}{|c|l|l|l|l|l|l|}
\hline
  \\[-4mm]
  \sola & $a$ @ 1 & $\mathit{null}$ @ 2 & $\overline{i}$ (ext) & $\overline{j}$ (1) & $k$ (2)\\
  \hline\hline  \\[-4mm]
  \solb & $\mathit{null}$ @ 1 & $z$ @ 2 & $\overline{i}$ (3)& $\overline{j}$ (ext) & $k$ (ext)\\
  \hline
\end{tabular}
\end{center}

\noindent
where $x$ @ $i$ denotes that literal $x$ is decided at level $i$, and $x\ (r)$
denotes that literal $x$ is propagated using a reason clause $r$, or due
to the other solver (if $r = \text{ext}$).
A conflict is hit in \solb in clause (4).
Inter-modular conflict analysis begins. \solb first asks for the reason for $k$,
which is clause (2) in \sola. This clause is copied to \solb. Note that unlike
\sms, clauses can move from $\sola$ to $\solb$. The new
conflict to be analyzed is $(i \lor j \lor j)$. Now the reason for $\overline{j}$
is asked of \sola. In this case, \sola cannot produce a clause over shared
variables to justify $\overline{j}$, so conflict analysis fails with
$\mathit{refineLit} = j$.
This causes \specsms to exit speculation mode and move to state $\langle P, D^0
\rangle$ and \sola must decide variable $j$ before speculating again. Note that in this
case either decision on $j$ results in $\langle\db_{\secs},
\db_{\mains}\rangle$ being sat. \qed
\end{example}

Case (2) is similar to what happens in \sms when a conflict is detected in
$\solb$. The reason for the conflict is below level $i$ which is below the
level of any decision of $\solb$. Since decision levels below $i$ are treated
as assumptions in $\solb$, calling $\analyzefinal$ in $\solb$ returns an
interface clause $c$ that blocks the current assignment in
$\sola$. The clause $c$ is added
to $\sola$. The solvers back-jump to the smallest decision level $j$ that makes $c$
an asserting clause in $\sola$. Finally, \specsms  moves to $\langle P, D^0 \rangle$.

\paragraph*{Validation} Case (3), like Case (1), is unique to \specsms. While
all clauses of $\solb$ are satisfied, the current assignment might not satisfy
all clauses of $\sola$. Thus, \specsms enters \emph{validation} by switching to
the configuration $\langle F, \decidemode^M \rangle$, where $M$ is the current
decision level. Thus, $\sola$ becomes active and starts deciding and
propagating. This continues, until one of two things happen: (3a) $\sola$
extends the assignment to satisfy all of its clauses, or (3b) a conflict that
cannot be resolved with inter-modular conflict analysis is found. In the case
(3a), \specsms transitions to $\langle F, F \rangle$ and declares that $\langle
\db_{\mains}, \db_{\secs} \rangle$ is sat. The case (3b) is handled exactly the
same as Case (1) -- the literal on the trail without a reason is stored in
$\mathit{refineLit}$, \specsms moves to $\langle P, D^0 \rangle$, backjumps to
the level in which speculation was started, and $\sola$ is forced to decide on
$\mathit{refineLit}$.

\begin{algorithm}[t]
  \footnotesize
\caption{The \textsc{Reason} method in modular SAT solvers inside \specsms}
\label{alg:reason}
\begin{algorithmic}[1]
  \Function{\textsc{Reason}}{$\mathit{lit}$}
   \If{$\mathit{reason}[\mathit{lit}] = \placeholder$}
    \State $c \leftarrow \mathit{other}.\analyzefinal(\mathit{lit})$
   \If{$\exists v \in c \cdot v \not\in I$}
    \State $\mathit{refineLit} \leftarrow \mathit{lit}$
    \State \Return $\placeholder$
   \EndIf
   \State $\textsc{AddClause}(c)$
   \State $\mathit{reason}[\mathit{lit}] \leftarrow c$
   \EndIf
   \State \Return $\mathit{reason}[\mathit{lit}]$
  \EndFunction
\end{algorithmic}
\end{algorithm}
 
\begin{theorem}
\label{thm:specsms}
    \specsms terminates. If it reaches the state $\langle F,F \rangle$, then $\db_\secs \land \db_\mains$ is satisfiable and the join of the trails of $\langle \solb, \sola\rangle$ is a satisfying assignment. If it reaches the state $\mathit{unsat}$, $\db_\secs \land \db_\mains$ is unsatisafiable.
\end{theorem}

 \section{Validation and interpolation}

In this section, we augment \specsms with an interpolation procedure. To this end, we first introduce modular DRUP proofs, which are generated from \specsms in a natural way. We then present an algorithm for extracting an interpolant from a modular trimmed DRUP proof in the spirit of~\cite{DBLP:conf/fmcad/GurfinkelV14}.

\subsection{DRUP proofs for modular SAT\label{sec:proofs}}

Modular DRUP proofs -- a form of clausal
proofs~\cite{DBLP:conf/date/GoldbergN03} -- extend (monolithic) DRUP
proofs~\cite{DBLP:conf/fmcad/HeuleHW13}. A DRUP
proof~\cite{DBLP:conf/fmcad/HeuleHW13} is a sequences of steps, where each step
either asserts a clause, deletes a clause, or adds a new Reverse Unit
Propagation (RUP) clause. Given a set of clauses $\Gamma$, a clause $\cls$ is an
RUP for $\Gamma$, written $\Gamma \trup \cls$, if $\cls$ follows from $\Gamma$
by unit propagation~\cite{DBLP:conf/isaim/Gelder08}. For a DRUP proof $\pi$, let
$\allasserted(\pi)$ denote all clauses of the asserted commands in $\pi$, then
$\pi$ shows that all RUP clauses of $\pi$ follow from $\allasserted(\pi)$. If
$\pi$ contains a $\bot$ clause, then $\pi$ certifies $\allasserted(\pi)$ is unsat.

A Modular DRUP proof is a sequence of clause addition and deletion steps,
annotated with indices \idx (\textsf{m} or \textsf{s}). Intuitively, steps with
the same index must be validated together (within the same module \idx), and
steps with different indices may be checked independently. The steps are:
\begin{enumerate}
\item $(\text{asserted},\idx, \cls)$ denotes that $\cls$ is asserted in \idx,
\item $(\text{rup}, \idx, \cls)$ denotes adding RUP clause \cls to $\idx$, \item $(\text{cp}(\mathit{src}), \mathit{dst},\cls)$ denotes copying a clause \cls from $\mathit{src}$ to $\mathit{dst}$, and
\item $(\text{del}, \idx, \cls)$ denotes removing clause \cls from $\idx$.
\end{enumerate}

We denote the prefix of length $k$ of a sequence of steps $\pi$ by $\pi^k$.
Given a sequence of steps $\pi$ and a formula index $\idx$, we use $\actcls(\pi,
\idx)$ to denote the set of active clauses with index $\idx$. Formally,
\begin{multline*}
\{\cls \mid
\exists c_j \in \pi \cdot \\ (c_j = (t, \idx ,\cls) \land (t = \text{asserted} \lor t = \rup \lor t = \cp(\_))) \\ {}\land \hphantom{a}
\neg \exists c_k \in \pi \cdot k > j \land c_k = (\del, \idx, \cls) \}
\end{multline*}

A sequence of steps $\pi = c_1,\ldots, c_n$ is a \emph{valid modular DRUP proof} iff for each $c_i \in \pi$:
\begin{enumerate}[leftmargin=*]
\item if $c_i = (\rup, \idx,\cls)$ then $\actcls(\pi^i, \idx) \trup \cls$,
\item if $c_i = (\cp(\idx), \_, \cls)$ then $\actcls(\pi^i, \idx) \trup \cls$, and
\item $c_{|\pi|}$ is either $(\rup, \mains, \bot)$ or $(\cp(\secs), \mains, \bot)$.
\end{enumerate}

Let $\allasserted(\pi, idx)$ be the set of all asserted clauses in~$\pi$ with index $idx$.
\begin{theorem}
  If $\pi$ is  a valid modular DRUP proof, then $\allasserted(\pi, \secs) \land \allasserted(\pi, \mains)$ is unsatisfiable.
\end{theorem}

Modular DRUP proofs may be validated with either one or two solvers.
To validate with one solver we convert the modular proof into a monolithic one
(i.e., where the steps are $\text{asserted}$, $\rup$, and $\del$).
Let \mbox{\textsc{modDRUP2DRUP}} be a procedure that given a modular DRUP proof
$\pi$, returns a DRUP proof $\pi'$ that is obtained from $\pi$ by (a) removing
$\idx$ from all the steps; (b) removing all $\cp$ steps; (c) removing all $\del$
steps. Note that $\del$ steps are removed for simplicity, otherwise it is
necessary to account for deletion of copied and non-copied clauses separately.

\begin{lemma}
  If $\pi$ is a valid modular DRUP proof then $\pi' =
  \textsc{modDRUP2DRUP}(\pi)$ is a valid DRUP proof.
\end{lemma}

Modular validation is done with two monolithic solvers working in lock step: $(\text{asserted},\cls, \idx)$ steps are added to the $\idx$ solver; $(\rup,\idx,\cls)$ steps are validated locally in solver $\idx$ using all active clauses (asserted, copied, and \rup); and for $(\cp(\mathit{src}), \mathit{dst}, \cls)$ steps, \cls is added to $\mathit{dst}$ but not validated in it, and \cls is checked to exist in the $\mathit{src}$ solver.

From now on, we consider only valid proofs.
We say that a (valid) modular DRUP proof $\pi$ is a proof of unsatisfiability of $\db_{\secs} \land \db_{\mains}$ if $\allasserted(\pi, \secs) \subseteq \db_{\secs}$ and  $\allasserted(\pi, \mains) \subseteq \db_{\mains}$ (inclusion here refers to the sets of clauses).

\specsms produces modular DRUP proofs by logging the clauses that are learnt, deleted, and copied between solvers.
Note that in \sms clauses may only be copied from $\solb$ to $\sola$, but in \specsms they might be copied in both directions.

\begin{theorem}
  Let $\db_{\secs}$ and $\db_{\mains}$ be two Boolean formulas s.t. $\db_{\secs} \land \db_{\mains} \models \bot$. \specsms produces a valid modular DRUP proof for unsatisfiability of $\db_{\secs} \land \db_{\mains}$.
\end{theorem}

\begin{algorithm}[t]
\footnotesize
  \caption{Trimming a modular DRUP proof\label{alg:smstrim}}
  \textbf{Input:} Solver instances $S_\secs$, $S_\mains$ with the empty clause on the trail, and a modular clausal proof $\pi = c_1,\ldots,c_n$. \\
  \textbf{Output:} A proof $\pi'$ s.t. all steps are core.
  \begin{algorithmic}[1]
    \State $\pi' = \emptyset$
    \State $M_\secs, M_\mains \leftarrow \{\bot\}, \emptyset$  \Comment{Relevant clauses}\For{$i=n \textbf{ to } 0$}
    \State \textbf{match} $c_i$ \textbf{with} $(\type,\idx,\cls)$
    \If{$\cls \not\in M_\idx$} \textbf{continue} 
    \EndIf
     \If{$\type = \del$}
     \State $S_{\idx}.\textsf{Revive}(\cls)$
     \State \textbf{continue}
     \EndIf
    \State $\pi'.\mathit{append}(c_i)$
\If{$\type = \rup$}
\State $S_\idx.$\Call{\sf\checkrup}{$\cls, M_\idx$}
\ElsIf{$\type = \cp(\mathit{src})$}
\State $S_{\idx}.\textsf{Delete}(\cls)$
\State $M_\mathit{src}.\mathit{add}(cls)$
\EndIf
\EndFor
    \State $\pi'.\mathit{reverse}()$
\Function{solver::{\sf\checkrup}}{$\cls, M$}
    \If{\textsf{IsOnTrail}(\cls)}\label{ln:ontrail}
    \State \textsf{UndoTrail}(\cls)\label{ln:undotrail}
    \EndIf
    \State \textsf{Delete}(\cls)
        \State $\textsf{SaveTrail}()$
        \State \textsf{Enqueue}$(\neg \cls)$
        \State $r \leftarrow \textsf{Propagate}()$
        \State $\textsf{ConflictAnalysis}(r, M)$
        \Comment{Updates $M$ with conflict clauses}
        \State $\textsf{RestoreTrail}()$
    \EndFunction
  \end{algorithmic}
\end{algorithm}

\noindent
\textit{Trimming modular DRUP proofs}. A step in a modular DRUP proof $\pi$ is \emph{core} if removing it invalidates $\pi$. Under this definition, $\del$ steps are never core since removing them does not affect validation.
\cref{alg:smstrim} shows an algorithm to trim modular DRUP proofs based on
backward validation. The input are two modular solvers $S_\mains$ and $S_\secs$
in a final conflicting state, and a valid modular DRUP proof $\pi =
c_1,\ldots,c_n$. The output is a trimmed proof $\pi'$ such that all steps of $\pi'$ are core.

We assume that the reader is familiar with MiniSAT~\cite{miniSAT} and use the following solver methods: \textsf{Propagate}, exhaustively applies unit propagation (UP) rule by resolving all unit clauses; \textsf{ConflictAnalysis} analyzes the most recent conflict and marks which clauses are involved in the conflict;
\textsf{IsOnTrail} checks whether a clause is an antecedent of a literal on the trail;
\textsf{Enqueue} enqueues one or more literals on the trail;
\textsf{IsDeleted}, \textsf{Delete}, \textsf{Revive} check whether a clause is deleted, delete a clause, and add a previously deleted clause, respectively;
\textsf{SaveTrail}, \textsf{RestoreTrail} save and restore the state of the trail.

\cref{alg:smstrim} processes the steps of the proof backwards, rolling back the states of the solvers.
$M_\idx$ marks which clauses were relevant to derive clauses in the current suffix of the proof.
While the proof is constructed through inter-modular reasoning, the trimming algorithm processes each of the steps in the proof completely locally.
During the backward construction of the trimmed proof, steps that include unmarked clauses are ignored (and, in particular, not added to the proof).
For each (relevant) \rup step, function \checkrup, using \textsf{ConflictAnalysis}, adds clauses to $M$.
\del steps are never added to the trimmed proof, but the clause is revived from the solver. For \cp steps, if the clause was marked, it is marked as used for the solver it was copied from and the step is added to the proof.
Finally, $\text{asserted}$ clauses that were marked are added to the trimmed proof.
Note that, as in~\cite{DBLP:conf/fmcad/GurfinkelV14}, proofs may be trimmed in different ways, depending on the strategy for \textsf{ConflictAnalysis}.

The following theorem states that trimming preserves validity of the proof.

\begin{theorem}
Let $\db_{\secs}$ and $\db_{\mains}$ be two formulas such that $\db_{\secs} \land \db_{\mains}\models \bot$. If $\pi$ is a modular DRUP proof produced by solvers $S_\secs$ and $\db_{\mains}$ for $\db_{\secs} \land \db_{\mains}$, then a trimmed proof $\pi'$ by \cref{alg:smstrim} is also a valid modular DRUP proof for $\db_{\secs} \land \db_{\mains}$.
\end{theorem}

\Cref{fig:drup} shows a trimmed proof after \specsms is executed on $\langle\psi_0,\psi_1\rangle$ such that $\psi_0 \eqdef ((s_1 \land \la_1)\limp s_2)) \land ((s_1 \land \neg \la_1)\limp s_2) \land ((s_3 \land \la_2)\limp s_4) \land ((s_3 \land \neg \la_2)\limp s_4)$ and $\psi_1 \eqdef (\neg s_1 \limp \lb_1) \land (\neg s_1 \limp \neg \lb_1) \land ((s_2 \land \lb_2)\limp s_3) \land ((s_2 \land \neg \lb_2)\limp s_3) \land (s_4 \limp \lb_3) \land (s_4 \limp \neg \lb_3))$.

\begin{figure}
  \footnotesize
\centering
\begin{tabular}{rlccl}
\textbf{seq} & \textbf{step} & \textbf{to} & \textbf{clause}\\
\hline
  1  & asserted & \mains & $\neg s_1 \limp \lb_1$\\
  2   & asserted & \mains & $\neg s_1 \limp \neg \lb_1$\\
  3   & asserted & \secs & ($s_1 \land \la_1)\limp s_2$\\
  4   & asserted & \secs & ($s_1 \land \neg \la_1)\limp s_2$\\
  5   & asserted & \mains & ($s_2 \land \lb_2)\limp s_3$\\
  6   & asserted & \mains & ($s_2 \land \neg \lb_2)\limp s_3$\\
  7   & asserted & \secs & ($s_3 \land \la_2)\limp s_4$\\
  8   & asserted & \secs & ($s_3 \land \neg \la_2)\limp s_4$\\
  9   & asserted & \mains & $s_4 \limp \lb_3$\\
 10   & asserted & \mains & $s_4 \limp \neg \lb_3$\\
 11   & \rup & \mains & $s_1$\\
 12   & \rup & \mains & $\neg s_4$\\
 13   & \rup & \mains & $s_2\limp s_3$\\
 14   & $\cp(\mains)$ & \secs & $s_2\limp s_3$ \\
 15   & \rup & \secs & $s_3\limp s_4$\\
 16   & \rup & \secs & $s_1\limp s_4$\\
 17   & $\cp(\secs)$ & \mains & $s_1 \limp s_4$ \\
 18   & \rup & \mains & $\bot$\\
 \hline
\end{tabular}
\caption{An example of a modular DRUP proof. Clauses are written in human-readable form as implications, instead of in the DIMACS format.}\label{fig:drup}
\end{figure}
 
\subsection{Interpolation\label{sec:itp}}

\begin{algorithm}[t]
  \footnotesize
  \caption{Interpolating a modular DRUP proof\label{alg:smsitp}.} \label{alg:specsmsitp}
  \textbf{Input:} Propositional formulas $\langle \db_0, \db_1\rangle$ \\
  \textbf{Input:} A modular trimmed DRUP proof $\pi = c_1,\ldots,c_n$ of unsatisfiability of $\db_0 \land \db_1$ \\
  \textbf{Output:} An interpolant $\itp$ s.t. $\db_0 \Rightarrow \itp$ and $\itp \land \db_1 \models \bot$
  \begin{algorithmic}[1]
    \State $S_\secs, S_\mains \gets \textsc{SAT\_Solver}()$
    \State $\itp \gets \top$
    \For{$i = 0 \textbf{ to } n$}
    \Match {$c_i$}
    \MCase{$(\text{asserted},\secs,\cls)$}
    \State $sup(\cls) \leftarrow \top$
    \EndMCase
    \MCase{$(\cp(\mains), \secs, \cls)$}
    \State $sup(\cls) \leftarrow \cls$\label{ln:bcp}
    \EndMCase
    \MCase{$(\rup,\secs,\cls)$}
    \State $M \gets \emptyset$
    \State $S_\secs.$\Call{\sf \checkrup}{$\cls, M$}
    \State $sup(\cls) \leftarrow \{sup(c)\mid c \in M\}$\label{ln:rup}
    \EndMCase
    \MCase{$(\cp(\secs), \mains, \cls)$}
    \State $\itp \leftarrow \itp \land (sup(\cls) \Rightarrow \cls)$
    \EndMCase
    \EndMatch
    \State $S_{c_i.idx}.add(\cls)$\EndFor
  \end{algorithmic}
\end{algorithm}
 Given a modular DRUP proof $\pi$ of unsatisfiability of $\db_{\secs} \land \db_{\mains}$, we give an algorithm to compute an interpolant of $\db_{\secs} \land \db_{\mains}$. 
For simplicity of the presentation, we assume that $\pi$ has no deletion steps; this is the case in trimmed proofs, but we can also adapt the interpolation algorithm to handle deletions by keeping track of active clauses.

Our interpolation algorithm relies only on the clauses copied between the modules. Notice that whenever a clause is copied from module $i$ to module $j$, it is
implied by all the clauses in $\Phi_{i}$ together with all the clauses that have
been copied from module $j$. We refer to clauses copied from $\sola$ to $\solb$
as \emph{backward} clauses and clauses copied from $\solb$ to $\sola$ as
\emph{forward} clauses. The conjunction of forward clauses is unsatisfiable with
$\sola$. This is because, in the last step of $\pi$, $\bot$ is added to $\sola$,
either through \rup or by \cp $\bot$ from $\solb$.  Since all the clauses in module \mains are implied by $\db_{\mains}$ together with forward clauses, this means that the conjunction of forward clauses is unsatisfiable with $\db_{\mains}$. In addition, all forward clauses were learned in module \secs, with support from backward clauses.
This means that every forward clause is implied by $\Phi_{\secs}$ together with the subset of the backward clauses used to derive it.
Intuitively, we should therefore be able to learn an interpolant with the structure: backward clauses imply forward clauses.

\cref{alg:specsmsitp} describes our interpolation algorithm. It traverses a modular DRUP proof forward. For each clause $\mathit{cls}$ learned in module \secs, the algorithm collects the set of backward clauses used to learn $\mathit{cls}$. This is stored in the $\sup$ datastucture --- a mapping from clauses to sets of clauses. Finally, when a forward clause $c$ is copied, it adds $\sup(c)\limp c$ to the interpolant.

\begin{example}
We illustrate our algorithm using the modular DRUP proof from \cref{fig:drup}. On the first
$\cp$ step~($\cp(\mains), \secs, s_2\limp s_3$), the algorithm assigns the $\sup$ for
clause $s_2\limp s_3$ as itself~(line~\ref{ln:bcp}). The first clause learnt in module \secs, ($\rup,\secs,s_3\limp
s_4$), is derived from just the clauses in module \secs and no backward clauses.
Therefore, after RUP, our algorithm sets $\sup(s_3\limp
s_4)$ to $\top$~(line~\ref{ln:rup}). The second clause learnt in module \secs, $s_1 \limp s_4$, is
derived from module \secs with the support of the backward clause $s_2\limp s_3$.
Therefore, $\sup(s_1\limp s_4) = \{s_2\limp s_3\}$. When this clause is copied
forward to module 1, the algorithm updates the interpolant to be $(s_2\limp s_3)
\limp (s_1\limp s_4)$. 
\qed
\end{example}

Next, we formalize the correctness of the algorithm. Let $L_B(\pi) = \{ \cls \mid (\cp(\mains), \secs, \cls) \in \pi \}$ be the set of clauses copied from module \mains to \secs and $L_F(\pi) = \{ \cls \mid (\cp(\secs), \mains, \cls) \in \pi \}$ be clauses copied from module \secs to \mains. From the validity of modular DRUP proofs, we have that:
\begin{lemma}\label{lm:itpinv}
    For any step $c_i = (\cp(\secs), \mains, \cls) \in \pi$, $(L_B(\pi^i) \land \Phi_{\secs}) \limp \cls$ and 
    for any step $c_j = (\cp(\mains), \secs, \cls) \in \pi$, $(L_F(\pi^j) \land \Phi_{\mains}) \limp \cls$.
\end{lemma}

For any clause $\cls$ copied from one module to the other, we use the shorthand $\id(\cls)$ to refer to the position of the copy command in the proof $\pi$.
That is, $\id(\cls)$ is the smallest $k$ such that $c_k = (\cp(i), j , \cls) \in \pi$.
The following is an invariant in a valid modular DRUP proof:
\begin{lemma}\label{lm:itpindinv}
    \[
    \forall \cls \in L_F(\pi) \cdot (\Phi_{\mains} \land (L_F(\pi^{\id(\cls)})) \limp L_B(\pi^{\id(\cls)}))
\]
\end{lemma}

These properties ensure that adding $L_B(\pi^{\id(\cls)}) \limp \cls$ for every forward clause $\cls$ results in an interpolant. \Cref{alg:smsitp} adds $(\sup(\cls) \limp \cls)$ as an optimization. Correctness is preserved since $\sup(\cls)$ is a subset of $L_B(\pi^{\id(\cls)})$ that together with $\db_{\secs}$ suffices to derive $\cls$ (formally, $\sup(\cls) \land \db_{\secs} \trup \cls$).

\begin{theorem}
Given a modular DRUP proof $\pi$ for $\db_{\secs} \land \db_{\mains}$, $\itp \eqdef \{\sup(c)\limp c \mid c
\in L_F(\pi)\}$ is an interpolant \mbox{for $\langle \db_{\secs}$, $\db_{\mains}\rangle$}.
\end{theorem}
\begin{proof}
    Since all copy steps are over interface variables, the interpolant is also over interface variables. By \cref{lm:itpinv} (and the soundness of $\sup$ optimization), $\db_{\secs} \limp \itp$. Next, we prove that $(\db_{\mains} \land \itp) \limp \bot$. From \Cref{lm:itpindinv}, we have that for all $c \in L_F(\pi)$, $(\db_{\mains} \land L_F(\pi^{\id(c)}))\limp \sup(c)$. Therefore, $(\db_{\mains} \land L_F(\pi^{\id(c)}) \land (\sup(c)\limp c)) \limp c$
\end{proof}

It is much simpler to extract interpolants from modular DRUP proofs then from
arbitrary DRUP proofs. This is not surprising since the interpolants capture
exactly the information that is exchanged between solvers. The interpolants are
not in CNF, but can be converted to CNF after extraction. 
 \section{Guiding \specsms via solver callbacks}
\label{sec:api}

As the reader may have noticed, deciding when to switch to speculative mode is non-trivial.
Heuristics may be implemented, as typically done in SAT solvers, but we consider this out of the scope of this paper.
Instead, we provide an interface for users to guide \specsms based on solver callbacks.
This scheme has been recently proven useful to guide SMT solving and to define custom theories in Z3~\cite{DBLP:conf/vmcai/BjornerEK23}. 

Users may provide a function \nextsplit to guide the solver in whether to speculate and over which variables to decide. \specsms calls \nextsplit whenever the next decision is about to be made.
\specsms expects \nextsplit to return  $\mathit{none}$ (default, in which case, the underlying heuristics are used) or a pair $(\mathit{ch\_mode},\mathit{Vars})$ where $\mathit{ch\_mode}$ is a Boolean that indicates a change to speculative mode and $\mathit{Vars}$ is a (possibly empty) set of variables to assign. 
To implement \nextsplit, users can ask the solver if a variable has been assigned using the \fixed function.
We illustrate the API with some examples.
\begin{example}\label{ex:experiments}
Consider modular queries of the following form:
$\langle\psi_{\invars}(\ell,\invars), \psi_{\sha}(\invars,\outvars)\rangle$, 
where $\ell$ is a 2-bit vector, \invars is a 512-bit vector (shared), \outvars is 160-bit vector. $\psi_{\invars}$ encodes that there are four possible messages: 
\begin{align*}
  \psi_{\invars} &\eqdef (\ell = 0 \land \invars = \msg_0) \lor (\ell = 1 \land \invars = \msg_1) \lor{} \\ 
  &\phantom{{}\eqdef{}} (\ell = 2 \land \invars = \msg_2) \lor (\ell = 3 \land \invars = \msg_3)
\end{align*}
and $\psi_{\sha}(\invars,\outvars)$ encodes the $\sha$ circuit together with some hash: 
\[
  \psi_{\sha} \eqdef (\sha_{circ}(in) \land \outvars = \mathit{shaVal})
\]
Roughly, the modular query $\langle\psi_{\invars}, \psi_{\sha}\rangle$ asks whether the $\sha$ of any $\msg_i$ is $\mathit{shaVal}$.
As we saw in \cref{sec:examples}, we are interested in using speculation in queries of this form 
in order to avoid the hard problem of inverting SHA-1 (as required by SMS, for example).
We can guide the solver with the function
$\nextsplit() \eqdef (\top, \ell)$.
\qed
\end{example}

Speculation is useful for such queries both in cases where the formulas are satisfiable and unsatisfiable.
If unsat, only the 4 inputs for SHA-1 specified by $\msg_i$ need to be considered, avoiding the expensive hash inversion problem. If sat, only two decisions on the bits of $\ell$ result in fully assigning \invars, which results again in just checking the hash.

\begin{example}
Next, consider a different form of modular query: $\langle\gamma_{\invars}(\ell,x,\invars), \gamma_{\sha}(\invars,x,\outvars)\rangle$, where $x$ is an $512$-bit vector, $\ell$ is a $160$-bit vector, $\chks_i$ are 512-bit vector, and the remaining variables are the same as in $\psi_\invars$ and $\psi_{\sha}$, and 
\begin{align*}
  \gamma_{\invars} &\eqdef 
   \sha_{circ}(x, \ell) \land{} \\
   &\phantom{{}\eqdef{}} \bigl( (\ell = \chks_0 \land \invars = \msg_0) \lor 
    (\ell = \chks_1 \land \invars = \msg_1) \lor{} \\
   &\phantom{{}\eqdef{}}  (\ell = \chks_2 \land \invars = \msg_2) \lor (\ell = \chks_3 \land \invars = \msg_3)\bigr)\\
  \gamma_{\sha} &\eqdef  (x = 1 \lor x = 4) \land  \sha_{\mathit{circ}}(\invars, \outvars) \land \outvars = \mathit{shaVal}
\end{align*}
This is an example 
where bi-directional search is necessary to efficiently solve the query. If deciding only on $\gamma_{\sha}$, we encounter the hard problem of inverting $\sha_{\mathit{circ}}$, if speculating directly in $\gamma_{\invars}$, we encounter the same problem, since an assignment for $x$ needs to be found, based on the four values for $\ell$. 

Accordingly, in this case, we are not interested in speculating immediately, but rather first decide on the value of $x$ in the main solver and then speculate. The following \nextsplit implementation captures this idea:
\begin{align*}
\nextsplit() \eqdef \ & \textbf{if} (\textbf{not\ } \fixed(x)) \ (\bot, x) \\
    & \textbf{else if}(\textbf{not\ } \fixed(\ell)) \  (\top, \ell) \\
    & \textbf{else} \ \textit{none}
\end{align*}
\qed
\end{example}
This example gives an intuition on which instances \specsms is better than \sms. Even if \sms is guided by \nextsplit, at least one inversion of $\sha_{\mathit{circ}}$ would have to be computed.

 \section{Implementation and Validation}
\label{sec:impl}

We have implemented \specsms (and \sms) inside the extensible SAT-solver of Z3~\cite{Z3}. For \sms, we simply
disable speculation. The code is publicly available on GitHub\footnote{\url{https://github.com/hgvk94/z3/tree/psms}.}. 

We have validated \specsms on a set of handcrafted benchmarks, based on~\cref{ex:experiments}, using the query $\langle \psi_{\invars}, \psi_{\sha} \rangle$.
In the first set of experiments, we check sat queries by generating one $\msg_i$ in $\psi_{\invars}$ that matches $\mathit{shaVal}$.
In the second set, we check unsat queries, by ensuring that no $\msg_i$  matches $\mathit{shaVal}$.
To evaluate performance, we make $\psi_{\sha}$ harder to solve by increasing the number of rounds of $\sha$ circuit encoded in the $\sha_{circ}$ clauses. We used \texttt{SAT-encoding}\footnote{Available at \url{https://github.com/saeednj/SAT-encoding}.} to generate the $\sha_{circ}$ with the different number of rounds (\texttt{SAT-encoding} supports 16 to 40 rounds).

We use the same guidance for both \sms and \specsms: $\nextsplit() \eqdef (\top, \ell)$. This means that once the secondary solver is active,
both \specsms and \sms branch on the $\ell$ variables first. However, \sms does not use speculation. Thus, it only switches to the secondary solver after finding a satisfying assignment to the main solver. 
In contrast, in \specsms, the secondary solver becomes active immediately due to speculation, and, accordingly, the search starts by branching on the $\ell$ variables.

\begin{table}[t]
\hspace{2mm}
    \begin{tabular}{crr}
        \toprule
        & \multicolumn{2}{c}{time (s) -- sat }\\
        \# rounds & \sms & \specsms\\
        \midrule
        16 & 1.96 & 0.41\\
        21 & -- & 0.66 \\
        26 & -- & 0.66 \\
        31 & -- & 0.81\\
        36 & -- & 1.01 \\
        40 & -- & 1.16 \\
        \bottomrule
    \end{tabular}
    \hfill
    \begin{tabular}{crr}
        \toprule
        & \multicolumn{2}{c}{time (s) -- unsat }\\
        \# rounds & \sms & \specsms\\
        \midrule
        16 & 0.77 & 0.65 \\
        21 & -- & 0.89 \\
        26 & -- & 0.91 \\
        31 & -- & 1.08 \\
        36 & -- & 1.45 \\
        40 & -- & 1.77 \\
        \bottomrule
    \end{tabular}\hspace{2mm}
    \caption{Solving time with a timeout of 600s.}\label{tab:experiments}
    \vspace{-0.2in}
\end{table}

Results  for each set of the queries are shown in \cref{tab:experiments}. Column ``\# rounds'' shows the number of $\sha$ rounds encoded in $\psi_{\sha}$.
The problems quickly become too hard for \sms. At the same time, \specsms solves all the queries quickly. Furthermore, the run-time of \specsms appears to grow linearly with the number of rounds.

The experiments show that prioritizing decisions on $\ell$, which is effective in \specsms with speculation, is ineffective in \sms. This is expected  since this guidance  becomes relevant to \sms only after the main solver guessed a satisfying assignment to $\psi_{\sha}$. This,  essentially amounts to \sms noticing the guidance only after inverting the $\sha$ circuit, which defeats the purpose of the guidance. As far as we can see, no other guidance can help \sms since there are no good variables to branch on in the main solver and \sms does not switch to the secondary solver until the main solver is satisfied.

 \section{Conclusion and Future Work}
\label{sec:conclusion}

Modular SAT-solving is crucial for efficient SAT-based unbounded Model Checking. Existing techniques, embedded in \pdr~\cite{DBLP:conf/vmcai/Bradley11} and extended in \sms~\cite{SMS}, trade the efficiency of the solver for the simplicity of conflict resolution. In this paper, we propose a new modular SAT-solver, called \specsms, that extends \sms with truly bi-directional reasoning. We show that it is provably more efficient than \sms (and, therefore, \pdr). We implement \specsms in Z3~\cite{Z3}, extend it with DRUP-style~\cite{DBLP:conf/fmcad/HeuleHW13} proofs, and proof-based interpolation. We believe this work is an avenue to future efficient SAT- and SMT-based Model Checking algorithms.  

In this paper, we rely on user callbacks to guide \specsms when to start speculation and (optionally) what variables to prefer in the decision heuristic. This is sufficient to show the power of bi-directional reasoning over uni-directional one of \pdr and \sms. However, this does limit the general applicability of \specsms. In the future, we plan to explore guiding speculation by the number of conflicts in the main solver, possibly using similar strategy used for guiding restarts in a modern CDCL SAT-solver~\cite{miniSAT}.

A much earlier version of speculation, called \emph{weak abstraction}, has been implemented in the \textsc{Spacer} Constrained Horn Clause (CHC) solver~\cite{DBLP:conf/cav/Gurfinkel22}. Since \textsc{Spacer} extends \pdr to SMT, the choice of speculation is based on theory reasoning. Speculation starts when the main solver is satisfied modulo some theories (e.g., Linear Real Arithmetic or Weak Theory of Arrays). Speculation often prevents \textsc{Spacer} from being stuck in any one SMT query. However, \textsc{Spacer} has no inter-modular propagation and no  \emph{refinement}. If \emph{validation} fails, speculation is simply disabled and the query is tried again without it. We hope that extending \specsms to theories can be used to make \textsc{Spacer} heuristics much more flexible and effective.

DPLL(T)-style~\cite{DBLP:conf/cav/GanzingerHNOT04} SMT-solvers can be seen as
modular SAT-solvers where the main module is a SAT solver and the secondary solver
is a theory solver (often EUF-solver that is connected to other theory solvers
such as a LIA solver). This observation credited as an intuition for
\sms~\cite{SMS}. In modern SMT-solvers, all decisions are made by the
SAT-solver. For example, if a LIA solver wants to split on a bound of a variable
$x$, it first adds a clause $(x \leq (b-1) \lor x \geq b)$, where $b$ is the
desired bound, to the SAT-solver and then lets the SAT-solver branch on the
clause. \specsms extends this interaction by allowing the secondary solver
(i.e., the theory solver) to branch without going back to the main solver.
Control is returned to the main solver only if such decisions tangle local
decisions of the two solvers. We hope that the core ideas of \specsms can be
lifted to SMT and allow more flexibility in the interaction between the
DPLL-core and theory solvers. 
\bibliography{biblio}
\bibliographystyle{abbrv}
\clearpage
\appendices
\section{\specsms as a set of Rules}

A configuration of a \specsms module is a \mbox{4-tuple}: $\langle \SM, \cc,
\db, \mdl \rangle$, where $\SM$ is the search mode~(either $\decidemode^i$, $P$,
or $F$), $\cc$ is either a (conflict) clause implied by $\db$ or the marker
$\mathit{none}$, $\db$ is the clause database, and $\mdl$ is the trail of the
solver.
The overall configuration of \specsms consists of configurations of both of its
solvers.
\begin{table*}
\setcellgapes{5pt}
\makegapedcells
\begin{tabular}{l l l}
$Initialize$ & $\langle P, no, \db_{\secs}, \emptylist\rangle, \langle \decidemode^0, no, \db_{\mains}, \emptylist \rangle$ & \\

$Prop_{\secs}$ & $\begin{multlined}[t][.4\linewidth]\langle \dontcare, no, \db_{\secs} : (C \lor \ell), \mdl_{\secs}\rangle, NC_{\mains} \Rightarrow\\[-2ex] \langle \dontcare, no, \db_{\secs}: (C\lor \ell), \mdl_{\secs}: \ell^{C\lor \ell}\rangle,  NC_{\mains}\end{multlined}$ & \makecell[l]{$\ell$ unassigned in $\mdl_{\secs}$, \\$\neg C \subseteq \mdl_{\secs}$}\\

$Prop_{\secs S}$ & $\begin{multlined}[t][.4\linewidth]\langle \dontcare, no, \db_{\secs}, \mdl_{\secs}:\ell^{X}\rangle, \langle \dontcare,no, \db_{\mains}, \mdl_{\mains}\rangle \Rightarrow\\[-2ex] \langle \dontcare,no, \db_{\mains}, \mdl_{\mains}:\ell^{X}\rangle, \langle \dontcare,no, \db_{\mains}, \mdl_{\mains}: \ell^{C\lor \ell}\rangle\end{multlined}$ & \makecell[l]{$\neg C\subseteq \mdl_{\mains}$, $\ell$, $C$ shared,\\ $C\lor \ell \in \db_{\secs}$, \\$\ell$ unassigned in $\mdl_{\mains}$}\\

$PropD_{\secs S}$ & $\begin{multlined}[t][.4\linewidth]\langle \dontcare, no, \db_{\secs}, \mdl_{\secs}:\ell^{\dontcare}\rangle, \langle \dontcare, no, \db_{\mains}, \mdl_{\mains}\rangle \Rightarrow\\[-2ex] \langle \dontcare, no, \db_{\secs}, \mdl_{\secs}:\ell^{\dontcare}\rangle, \langle \dontcare, no, \db_{\mains}, \mdl_{\mains}: \ell^{\bot}\rangle\end{multlined}$ & \makecell[l]{$\ell$ is shared, $\ell$ is unassigned in $\mdl_{\mains}$,\\ there does not exists a clause $C \lor \ell \in \db_{\mains}$ s.t. \\ $C$ is shared and $\neg C \subseteq \mdl_{\secs}$} \\

$Conflict_{\secs}$ & $\begin{multlined}[t][.4\linewidth]\langle \dontcare,no, \db_{\secs} : C, \mdl_{\secs}\rangle, NC_{\mains} \Rightarrow\\[-2ex] \langle\dontcare, C, \db_{\secs}: C, \mdl_{\secs}\rangle, NC_{\mains}\end{multlined}$& \makecell[l]{$\neg C \subseteq \mdl_{\secs}$} \\

$Explain_{\secs}$ & $\begin{multlined}[t][.4\linewidth]\langle \dontcare, (\neg \ell \lor C), \db_{\secs}, \mdl_{\secs}: \ell^{(\ell \lor D)}\rangle, NC_{\mains} \Rightarrow\\[-2ex] \langle \dontcare, (D \lor C), \db_{\secs}, \mdl_{\secs}: \ell^{(\ell \lor D)}\rangle, NC_{\mains} \end{multlined}$ & \\

\end{tabular}
\caption{Rules independent of solving modes. For each rule, except $Initialize$, there is a symmetrical rule to update Solver $\mains$. These rules are not shown here for brevity.$NC_B = \langle \dontcare, no, \db_{\mains}, \mdl\rangle$}
\end{table*}

\begin{table*}
\setcellgapes{5pt}
\makegapedcells

\begin{tabular}{l l l}

  $Decide_{\mains}$ & $\begin{multlined}[t][.4\linewidth]\langle \dontcare, no, \db_{\secs}, \mdl_{\secs}\rangle, \langle \decidemode, no, \db_{\mains}, \mdl_{\mains}\rangle \Rightarrow\\[-2ex] \langle \dontcare, no, \db_{\secs}, \mdl_{\secs}:\nulllit\rangle, \langle \decidemode, no, \db_{\mains}, \mdl_{\mains}:\ell\rangle\end{multlined}$ & $\ell$ unassigned in $\mdl_{\mains}$, $\ell$ not shared \\
  $Decide_{\mains S}$ & $\begin{multlined}[t][.4\linewidth]\langle \dontcare, no, \db_{\secs}, \mdl_{\secs}\rangle, \langle \decidemode, no, \db_{\mains}, \mdl_{\mains}\rangle \Rightarrow\\[-2ex] \langle \dontcare, no, \db_{\secs}, \mdl_{\secs}:\nulllit,\ell^{\bot}\rangle, \langle \decidemode, no, \db_{\mains}, \mdl_{\mains}:\ell\rangle\end{multlined}$ & $\ell$ unassigned in $\mdl_{\mains}$, $\ell$ shared \\
  $Reason_{\secs}$ & $\begin{multlined}[t][.4\linewidth]\langle \dontcare, no, \db_{\secs}, \mdl_{\secs}:\ell^{(\ell\lor C)}\rangle, NC_{\mains} \Rightarrow\\[-2ex] \langle \dontcare, (\ell\lor C), \db_{\secs}, \mdl_{\secs}:\ell^{(\ell\lor C)}\rangle, NC_{\mains}\end{multlined}$ & \makecell[l]{$\ell$ is shared}\\
  
$Explain_{\mains \secs}$ & $\begin{multlined}[t][.4\linewidth]\langle P, C, \db_{\secs}, \mdl_{\secs}\rangle, \langle \decidemode, no, \db_{\mains}, \mdl_{\mains}\rangle \Rightarrow\\[-2ex] \langle P, no, \db_{\secs}, \mdl_{\secs}\rangle, \langle \decidemode, C, \db_{\mains}, \mdl_{\mains}\rangle \end{multlined}$ & $C$ is shared, $C\in \db_{\mains}$\\

$Learn_{\mains}$ & $\begin{multlined}[t][.4\linewidth]NC_{\secs}, \langle \decidemode, C, \db_{\mains}, \mdl_{\mains}\rangle \Rightarrow\\[-2ex] NC_{\secs}, \langle \decidemode, C, \db_{\mains}: C, \mdl_{\mains}\rangle\end{multlined}$ & $C \not \in \db_{\mains}$\\

$Learn_{\mains \secs}$ & $\begin{multlined}[t][.4\linewidth]\langle \dontcare, C, \db_{\secs}, \mdl_{\secs}\rangle, \langle \dontcare, no, \db_{\mains}, \mdl_{\mains}\rangle \Rightarrow\\[-2ex] \langle \dontcare, C, \db_{\secs}, \mdl_{\secs}\rangle, \langle \dontcare, C, \db_{\mains}:C, \mdl_{\mains}\rangle \end{multlined}$ & \makecell[l]{$C$ is shared\\$C\not\in \db_{\mains}$} \\

  $Backtrack_{\mains}$ & $\begin{multlined}[t][.4\linewidth]\langle \dontcare, no, \db_{\secs}, M^k_{\secs}\rangle, \langle \decidemode^{i, j}, no, \db_{\mains}, M^k_{\mains}\rangle \Rightarrow\\[-2ex] \langle \dontcare, no, \db_{\secs}, \mdl_{\secs}^{k - 1}\rangle, \langle \decidemode^{i, j}, no, \db_{\mains}, \mdl_{\mains}^{k -1}\rangle \end{multlined}$ & \makecell[l]{$k > j$} \\
            
  $Backjump_{\mains}$ & $\begin{multlined}[t][.4\linewidth]\langle P, no, \db_{\secs}, \mdl_{\secs}\rangle, \langle \decidemode, C, \db_{\mains}, \mdl_{\mains}\rangle \Rightarrow\\[-2ex] \langle P, no, \db_{\secs}, \mdl_{\secs}^i\rangle, \langle \decidemode, no, \db_{\mains}, \mdl_{\mains}^i\rangle \end{multlined}$ & \makecell[l]{$i$ is the second highest decision level in $\neg C$}\\

  $Fail_{\mains}$ &  $\begin{multlined}[t][.4\linewidth]\langle\dontcare,  no, \db_{\secs}, \mdl_{\secs}\rangle, \langle \dontcare, \neg C, \db_{\mains}: C, \mdl_{\mains}\rangle \Rightarrow unsat\end{multlined}$ & $\neg C$ is decided at level 0\\
\end{tabular}
\caption{Rules when $\sola$ makes decisions and $\solb$ does not. $\nulllit$ denotes a null literal, a dummy literal that increments the decision level without assigning values to any variable.}
\label{tab:brules}
\end{table*}

\begin{table*}
\setcellgapes{5pt}
\makegapedcells
\begin{tabular}{l l l}
  $FA_\mains$ & $\begin{multlined}[t][.4\linewidth]\langle P, no, \db_{\secs}, \mdl^i_{\secs}\rangle, \langle \decidemode, no, \db_{\mains}, \mdl^i_{\mains}\rangle \Rightarrow\\[-2ex] \langle \decidemode^i, no, \db_{\secs}, \mdl_{\secs}^i\rangle, \langle F, no, \db_{\mains}, \mdl_{\mains}^i\rangle \end{multlined}$ & \makecell[l]{$\mdl^i_{\mains}$ is a full satisfying assignment to $\db_{\mains}$} \\
\end{tabular}
\caption{Rule to enter fin mode}
\label{tab:finen}
\end{table*}

\begin{table*}
\setcellgapes{5pt}
\makegapedcells
\begin{tabular}{l l l}
  $SPECM_{\mains}$ & $\begin{multlined}[t][.4\linewidth]\langle P, no, \db_{\secs}, \mdl^i_{\secs}\rangle, \langle \decidemode, no, \db_{\mains}, \mdl^i_{\mains}\rangle \Rightarrow\\[-2ex] \langle \decidemode^i, no, \db_{\secs}, \mdl_{\secs}^i\rangle, \langle P, no, \db_{\mains}, \mdl_{\mains}^i\rangle \end{multlined}$ & \makecell[l]{No clause is unit in $\db_{\mains}$ under $\mdl_{\mains}$ and $\db_{\secs}$ under $\mdl_{\secs}$} \\
\end{tabular}
\caption{Rule to enter speculation}
\label{tab:lamen}
\end{table*}

\begin{table*}
\setcellgapes{5pt}
\makegapedcells
\begin{tabular}{l l l}
  $Decide_{\secs}$ & $\begin{multlined}[t][.4\linewidth]\langle \decidemode, no, \db_{\secs}, \mdl_{\secs}\rangle, \langle P/F, no, \db_{\mains}, \mdl_{\mains}\rangle \Rightarrow\\[-2ex] \langle \decidemode, no, \db_{\secs}, \mdl_{\secs}:\ell\rangle, \langle P/F, no, \db_{\mains}, \mdl_{\mains}:\nulllit\rangle\end{multlined}$ & $\ell$ unassigned in $\mdl_{\secs}$, $\ell$ not shared \\
  $Decide_{\secs S}$ & $\begin{multlined}[t][.4\linewidth]\langle \decidemode, no, \db_{\secs}, \mdl_{\secs}\rangle, \langle P, no, \db_{\mains}, \mdl_{\mains}\rangle \Rightarrow\\[-2ex] \langle S, no, \db_{\secs}, \mdl_{\secs} :\ell\rangle, \langle P, no, \db_{\mains}, \mdl_{\mains}:\nulllit,\ell^{\bot}\rangle \end{multlined}$ & $\ell$ unassigned in $\mdl_{\secs}$, $\ell$ shared \\
  $Reason_{\mains}$ & $\begin{multlined}[t][.4\linewidth] NC_{\secs}, \langle \dontcare, no, \db_{\mains}, \mdl_{\mains}:\ell^{(\ell\lor C)}\rangle \Rightarrow\\[-2ex]  NC_{\secs}, \langle \dontcare, (\ell\lor C), \db_{\mains}, \mdl_{\mains}:\ell^{(\ell\lor C)}\rangle\end{multlined}$ & \makecell[l]{$\ell$ is shared}\\
  $Explain_{\secs \mains}$ & $\begin{multlined}[t][.4\linewidth]\langle \decidemode, no, \db_{\secs}, \mdl_{\secs}\rangle, \langle P/F, C, \db_{\mains}, \mdl_{\mains}\rangle \Rightarrow\\[-2ex] \langle \decidemode, C, \db_{\secs}, \mdl_{\secs}\rangle, \langle P/F, no, \db_{\mains}, \mdl_{\mains}\rangle \end{multlined}$ & $C$ is shared, $C\in \db_{\secs}$\\

$Learn_{\secs}$ & $\begin{multlined}[t][.4\linewidth] \langle \decidemode, C, \db_{\secs}, \mdl_{\secs}\rangle, NC_{\mains} \Rightarrow\\[-2ex] \langle \decidemode, C, \db_{\secs}: C, \mdl_{\secs}\rangle , NC_{\mains} \end{multlined}$& $C \not \in \db_{\secs}$\\

  $Learn_{\secs \mains}$ & $\begin{multlined}[t][.4\linewidth]\langle \decidemode^i, no, \db_{\secs}, \mdl_{\secs}\rangle, \langle P, C, \db_{\mains}, \mdl_{\mains}\rangle \Rightarrow\\[-2ex] \langle \decidemode^i, no, \db_{\secs}:C, \mdl_{\secs}\rangle, \langle P, C, \db_{\mains}, \mdl_{\mains}\rangle \end{multlined}$ & \makecell[l]{$C$ is shared\\$C\not\in \db_{\secs}$} \\

  $Backtrack_{\secs}$ & $\begin{multlined}[t][.4\linewidth]\langle \decidemode^i, C, \db_{\secs}, \mdl^k_{\secs}\rangle, \langle P/F, no, \db_{\mains}, \mdl^k_{\mains}\rangle \Rightarrow\\[-2ex] \langle \decidemode^i, no, \db_{\secs}, \mdl_{\secs}^{k - 1}\rangle, \langle P/F, no, \db_{\mains}, \mdl_{\mains}^{k -1}\rangle \end{multlined}$ & \makecell[l]{$k > i$} \\
  $Backjump_{\secs}$ & $\begin{multlined}[t][.4\linewidth]\langle \decidemode^i, C, \db_{\secs}, \mdl_{\secs}\rangle, \langle P/F, no, \db_{\mains}, \mdl_{\mains}\rangle \Rightarrow\\[-2ex] \langle \decidemode^i, no, \db_{\secs}, \mdl_{\secs}^k\rangle, \langle P/F, no, \db_{\mains}, \mdl_{\mains}^k\rangle \end{multlined}$ & \makecell[l]{$j$ is the second highest decision level in $\neg C$, \\ $k = max(j, i)$} \\
\end{tabular}
\caption{Rules when $\solb$ is making decisions and $\sola$ is not. All rules are symmetric to the ones in \cref{tab:brules}. The only difference is that there is no counterpart for $Fail_{\mains}$. This rule is presented in \cref{tab:lame}}
\label{tab:arules}
\end{table*}

\begin{table*}
\setcellgapes{5pt}
\makegapedcells
\begin{tabular}{l l l}
  $Explain\bot_{\secs}$ & $\begin{multlined}[t][.4\linewidth]\langle \decidemode^i, (\neg\ell \lor C), \db_{\secs}, \mdl_{\secs}:\ell^{\bot}\rangle, \langle P/F, no, \db_{\mains}, \mdl_{\mains}:\ell^{\dontcare}\rangle \Rightarrow\\[-2ex] \langle P, no, \db_{\secs}, \mdl_{\secs}^i\rangle, \langle \decidemode^0, no, \db_{\mains}, \mdl_{\mains}^i:\neg \ell\rangle \end{multlined}$ & \makecell[l]{$\ell$ is at a higher decision level\\ than all other literals in $C$} \\
  $Explain\bot_{\secs \mains}$ & $\begin{multlined}[t][.4\linewidth]\langle \decidemode^i, no, \db_{\secs}, \mdl_{\secs}:\ell^\dontcare\rangle, \langle P, C, \db_{\mains}, \mdl_{\mains}\rangle \Rightarrow\\[-2ex] \langle P, no, \db_{\secs}, \mdl_{\secs}^i\rangle, \langle \decidemode^0, no, \db_{\mains}, \mdl_{\mains}^i:\neg \ell\rangle \end{multlined}$ & \makecell[l]{$\exists m \in C$ s.t $\neg m$ is \\a local decision in $\sola$,\\ $\ell$ is shared, $\ell$ unassigned in $\mdl^i_{\mains}$} \\
  $FA_{A}$ & $\begin{multlined}[t][.4\linewidth]\langle \decidemode^i, no, \db_{\secs}, M^j_{\secs}\rangle, \langle P/F, no, \db_{\mains}, \mdl^j_{\mains}\rangle \Rightarrow\\[-2ex] \langle F, no, \db_{\secs}, \mdl^j_{\secs}\rangle, \langle \decidemode^{i,j}/F, no, \db_{\mains}, \mdl^j_{\mains}\rangle \end{multlined}$ & \makecell[l]{$\mdl_{\secs}$ is a full satisfying \\ assignment to $\db_{\secs}$} \\
  $Fail_{\secs}$ &  $\begin{multlined}[t][.4\linewidth]\langle \decidemode^i, \neg C, \db_{\secs}: C, \mdl_{\secs}\rangle, \langle P, no, \db_{\mains}, \mdl_{\mains}\rangle \Rightarrow\\[-2ex] \langle \decidemode^i, \neg C, \db_{\secs}: C, \mdl_{\secs}\rangle, \langle P, \top, \db_{\mains}:\bot, \mdl_{\mains}\rangle\end{multlined}$ & $\neg C$ is decided at level $0$\\
\end{tabular}
\caption{All rules to exit speculation.}
\label{tab:lame}
\end{table*}

\begin{table*}
\setcellgapes{5pt}
\makegapedcells
\begin{tabular}{l l l}
  $BackjumpV_{\mains}$ & $\begin{multlined}[t][.4\linewidth]\langle F, no, \db_{\secs}, \mdl_{\secs}\rangle, \langle \decidemode^{i,j}, C, \db_{\mains}, \mdl_{\mains}\rangle \Rightarrow\\[-2ex] \langle F, no, \db_{\secs}, \mdl_{\secs}^k\rangle, \langle \decidemode^{i, j}, no, \db_{\mains}, \mdl_{\mains}^k\rangle \end{multlined}$ & \makecell[l]{$d$ is the second highest decision level in $\neg C$, \\ $d \geq i$, $k = max(j, d)$} \\
\end{tabular}
\caption{Adaptation to the backjump rule for validate mode}
\label{tab:val}
\end{table*}

\begin{table*}
\setcellgapes{5pt}
\makegapedcells
\begin{tabular}{l l l}
  $Explain\bot_{\mains}$ & $\begin{multlined}[t][.4\linewidth]\langle F, no, \db_{\secs}, \mdl_{\secs}:\ell^{\dontcare}\rangle, \langle \decidemode^{i, j}, (\neg\ell \lor C), \db_{\mains}, \mdl_{\mains}:\ell^{\bot}\rangle  \Rightarrow\\[-2ex] \langle P, no, \db_{\secs}, \mdl_{\secs}^i\rangle, \langle \decidemode^0, no, \db_{\mains}, \mdl_{\mains}^i:\neg \ell\rangle \end{multlined}$ & \makecell[l]{$\ell$ is at a higher decision level\\ than all literals in $C$}\\
  $BackjumpVE_{\mains}$ & $\begin{multlined}[t][.4\linewidth]\langle F, no, \db_{\secs}, \mdl_{\secs}\rangle, \langle \decidemode^{i,j}, C, \db_{\mains}, \mdl_{\mains}\rangle \Rightarrow\\[-2ex] \langle P, no, \db_{\secs}, \mdl_{\secs}^k\rangle, \langle \decidemode^0, no, \db_{\mains}, \mdl_{\mains}^k\rangle \end{multlined}$ & \makecell[l]{$k$ is the second highest \\decision level in $\neg C$, $k < i$} \\
  $FAV_{\mains}$ & $\begin{multlined}[t][.4\linewidth]\langle F, no, \db_{\secs}, \mdl_{\secs}\rangle, \langle \decidemode^{i,j}, no, \db_{\mains}, \mdl_{\mains}\rangle \Rightarrow\\[-2ex] \langle F, no, \db_{\secs}, \mdl_{\secs}\rangle, \langle F, no, \db_{\mains}, \mdl_{\mains}\rangle \end{multlined}$ & \makecell[l]{$\mdl_{\mains}$ is a full satisfying \\ assignment to $\db_{\mains}$} \\
\end{tabular}
\caption{All rules to transition out of validate mode}
\label{tab:vale}
\end{table*}
 \end{document}